\documentclass[aps,twocolumn,showpacs,pra,superscriptaddress,floatfix,longbibliography]{revtex4-1}

\usepackage{amsmath,amssymb,amsmath,amsthm,bm,graphicx}
\usepackage{natbib}
\usepackage{color}
\usepackage{tabularx}
\usepackage{wasysym}
\usepackage{dcolumn}
\usepackage{subfigure}
\usepackage{bbold}
\usepackage{lipsum}

\newcommand{\Z}{{\mathbb{Z}}}

\newcommand{\R}{{\mathbb{R}}}

\DeclareMathOperator{\Tr}{Tr}
\newcommand{\bd}[1]{\boldsymbol{ #1 }}
\newcommand{\bra}[1]{\mbox{$\langle #1 |$}}
\newcommand{\ket}[1]{\mbox{$| #1 \rangle$}}

\renewcommand{\H}{\mathcal{H}}

\newtheorem{thm}{Theorem}
\newtheorem{lem}[thm]{Lemma}

 \definecolor{darkgreen}{rgb}{0,0.5,0} % VM
 \usepackage{color}
 \usepackage[normalem]{ulem}
 \usepackage{cancel}
\begin{document}

\title{Role of the pair potential for the saturation of generalized Pauli constraints}

\author{\"Ors Legeza}
\affiliation{Wigner Research Centre for Physics, Hungarian Academy of Sciences, Konkoly-Thege Mikl\'os \'ut 29-33, 1121 Budapest, Hungary}

\author{Christian Schilling}
\email{christian.schilling@physics.ox.ac.uk}
\affiliation{Clarendon Laboratory, University of Oxford, Parks Road, Oxford OX1 3PU, United Kingdom}

\date{\today}

\begin{abstract}
The dependence of the (quasi-)saturation of the generalized Pauli constraints on the pair potential is studied for ground states of few-fermion systems. For this, we consider spinless fermions in one dimension which are harmonically confined and interact by pair potentials of the form $|x_i-x_j|^s$ with $-1 \leq s\leq 5$. We use the Density Matrix Renormalization Group approach and large orbital basis to achieve the convergence on more than ten digits of both the variational energy and the natural occupation numbers. Our results confirm that the conflict between energy minimization and fermionic exchange symmetry results in a universal and non-trivial quasi-saturation of the generalized Pauli constraints (\emph{quasipinning}), implying tremendous structural simplifications of the fermionic ground state for all $s$.
Those numerically exact results are complemented by an analytical study based on a self-consistent perturbation theory which we develop for this purpose. The respective results for the weak coupling regime eventually elucidate the singular behaviour found for the specific values $s=2,4,\ldots$, resulting in an extremely strong quasipinning.
\end{abstract}


% 03.65.-w Quantum mechanics
% 05.30.Fk Fermion Gas, quantum statistical mechanics
% 31.15.ac High-precision calculations for few-electron (or few-body) atomic systems

\maketitle

\section{Introduction}\label{sec:intro}
In recent years, quantum information theoretical concepts began to play a more important role in the description and understanding of
quantum many-body systems. One of the prime examples reflecting the progress at that exciting interface and its fruitful future prospects is given by Klyachko's breakthrough result a few years ago: It has been shown that Pauli's original exclusion principle \cite{Pauli1925} --- despite its long-term success on all physical length scales --- is incomplete. Indeed, as already suggested by first studies \cite{Borl1972,Rus2}, the fermionic exchange symmetry has been found \cite{Kly2,Kly3,Altun} to imply more restrictive constraints on the one-particle picture, independent of the underlying Hamiltonian, rendering Pauli's original principle obsolete. To be more precise, these so-called \emph{generalized Pauli constraints} (GPC) take the form of linear conditions,
\begin{equation}\label{GPC}
  D_j(\vec{\lambda}) \equiv \kappa_j^{(0)} + \sum_{k=1}^d \kappa_j^{(k)} \lambda_k\geq 0\,,
\end{equation}
on the decreasingly-ordered \emph{natural occupation numbers}  $\vec{\lambda}\equiv (\lambda_k)_{k=1}^d$, the eigenvalues of the one-particle reduced density operator $\rho_1\equiv N \mbox{Tr}_{N-1}[\ket{\Psi_N}\!\bra{\Psi_N}]$. Here, $\ket{\Psi_N}\in \wedge^N[\mathcal{H}_1^{(d)}]$ is the $N$-fermion quantum state, where the one-particle Hilbert space $\mathcal{H}_1^{(d)}$ has dimension $d$ and $j=1,\ldots,\nu^{(N,d)}$, $\kappa_j^{(k)}\in \Z$. For each setting $(N,d)$, the finite set of GPCs defines a convex \emph{polytope} $\mathcal{P}\subset \R^d$, a subset of the \emph{Pauli simplex} $1\geq \lambda_1\geq \ldots \geq \lambda_d\geq 0$. Moreover, the values $D_j(\vec{\lambda})$ coincide up to a prefactor with the $l_1$-distances of $\vec{\lambda}$ to the respective polytope facets $F_j$ corresponding to $D_j\equiv0$ \cite{CS2016a}.

In complete analogy to Pauli's original exclusion principle, the potential physical relevance of the GPCs would primarily be based on their saturation in concrete systems. Such pinning would then reduce the complexity of the system's quantum state and would define \emph{``a new physical entity with its own dynamics and kinematics''} \cite{Kly1} (see also \cite{CSQuasipinning,Stability}). Moreover, the geometrical structure underlying the GPCs suggests a natural hierarchy of extensions of the Hartree-Fock ansatz allowing one to systematically capture static correlations in few-body quantum systems \cite{Stability}.

In an analytical study \cite{CS2013} of three harmonically interacting spinless fermions in a one-dimensional harmonic trap it has been shown that the GPCs are saturated up to a very small correction of the order eight in the dimensionless coupling, $D\propto \kappa^8$. A succeeding extended study of such harmonic models \cite{Ebler,CSthesis,CS2016a,CS2016b,CSQ,FTthesis} has confirmed, by varying the particle number, spatial dimension, degree of spin-polarization and coupling strength that this striking \emph{quasipinning}-effect has a physical origin. It namely emerges from the conflict between energy minimization and fermionic exchange symmetry \cite{CS2016b}. The intensive ongoing debate discussing the physical relevance of the GPCs in more realistic systems, as, e.g., atoms and molecules (see \cite{Kly1,BenavLiQuasi,Kly5,Mazz14,CSthesis,MazzOpen,CSQuasipinning,BenavQuasi2,RDMFT,Mazzagain,Alex,CS2015Hubbard,RDMFT2,CBthesis,NewMazziotti,Mazz16,CBcorr,MazzOpen2,CSQChem,TomaszDoub}
and references therein) has been hampered, however, due to a couple of reasons: First, most of the numerical studies of realistic systems were based so far on very small active spaces of 3-5 orbitals and may therefore fail to accurately capture the true physical situation. Second, it has only been realized very recently that (quasi)pinning is in some cases trivial in the sense that it is a mere consequence of spatial symmetries \cite{BenavQuasi2,CS2015Hubbard,CBthesis} or may just follow from an (approximate) saturation of the Pauli constraints \cite{Mazz14,CSQ}.

In the present work, we revisit the quasipinning phenomenon by addressing all the concerns mentioned above. We namely use state-of-the-art numerical methods to calculate the natural occupation numbers to a high precision and use a concise quantitative measure to distinguish trivial from non-trivial (quasi)pinning. The main aim is thus to estimate the \emph{genuine} relevance that GPCs have in few-fermion quantum systems. By considering one-dimensional harmonic trap systems of fully-polarized fermions, we ensure that possible (quasi)pinning is neither an artifact of orbital-symmetries nor of spin-symmetries. On the other hand, freezing degenerate angular
and spin degrees of freedom and considering a steep external trap reduces the size of the underlying active space and thus allows us to eventually perform a fully conclusive analysis of the GPCs' relevance in those systems. This is quite different in atoms, where the Coulomb interaction between the electrons manifests itself in a wave function cusp which entails the presence of dynamic correlations. While those dynamic correlations typically do not change the qualitative physical behavior, their recovering with high precision requires, however, much larger active spaces for which the generalized Pauli constraints are not known yet. We therefore believe that our succeeding analysis will give us a rather good idea of the genuine relevance that GPCs have in few-fermion quantum systems.

The Hamiltonians at hand take the form
\begin{equation}\label{ham}
\hat{H} = \sum_{j=1}^N \left(\frac{\hat{p}_j^2}{2m}+\frac{1}{2}m\omega^2 \hat{x}_j^2\right)\,+\,K\,\sum_{j<k}\,|\hat{x}_j-\hat{x}_k|^s\,.
\end{equation}
%We restrict \eqref{ham} to $s>-1$ since for $s\leq -1$ the system (at least for spinful fermions) would have no finite ground state energy for negative couplings $K$.
By using the natural length scale of the external harmonic trap, $l\equiv \sqrt{\hbar/m\omega^2}$, we introduce the dimensionless coupling parameter
\begin{equation}\label{kappa}
\kappa \equiv K/ m \omega^2 l^{2-s}\,.
\end{equation}
This is equivalent to just setting  $m\equiv \omega \equiv \hbar \equiv 1$. Besides the general motivation above, there are further important reasons for choosing the family of Hamiltonians of the form \eqref{ham} rather than the electronic Hamiltonians studied in quantum chemistry: The remarkable recent progress in quantum optics allows one to simulate with ultracold gases an increasing variety of physical systems and models which high flexibility and control. For instance, in contrast to the electrons in atoms and molecules, the interaction between the ultracold fermionic atoms can be tuned at the Feshbach resonance \cite{BlochReview,Feshbach,ultracoldBook1}. Furthermore, by departing from the dilute gas regime, the effective pair interaction between the fermionic atoms is typically described by a Lennard-Jones-type potential, and thus differs not that much anymore from \eqref{ham}. The future prospects of  quantum simulation in general and the proposed experimental realization of the (quasi)pinning-effect in systems of ultracold fermionic atoms in particular (see, e.g., the expected ``transparency effect'' \cite{CS2015Hubbard}) provide further compelling reasons for studying the GPCs in harmonic trap systems rather than in weakly-correlated few-electron atoms.

In the following section, we develop  a self-consistent perturbation theoretical approach to determine the leading order behavior of the NONs in the regime of weak coupling. In Section \ref{sec:DMRG} we then explain how to obtain the numerically exact ground state of \eqref{ham} for finite coupling by using the Density Matrix Renormalization Group (DMRG) approach \cite{White-1992b}. Eventually, we use those results to discuss in Section \ref{sec:sgrid} the relevance of the GPCs for different Hamiltonians \eqref{ham} by considering about 100 different $s$-values.

\section{Analytical approach for weak coupling}\label{sec:pt}
In this section we elaborate on a perturbation theoretical approach to determine the leading order behaviour of the minimal distance $D(\vec{\lambda}(\kappa))$ of $\vec{\lambda}$ to the polytope boundary for the ground state $\ket{\Psi(\kappa)}$ of a \emph{general}
$N$-particle Hamiltonian,
\begin{equation}\label{PTham}
\hat{H}(\kappa) = \hat{H}^{(0)}+\kappa \hat{V}\,,
\end{equation}
in the regime of weak coupling. Here, $\hat{H}^{(0)}$ denotes the one-particle Hamiltonian, $\hat{V}$ the pair interaction and the Hamiltonian acts on the $N$-fermion Hilbert space $\wedge^N[\mathcal{H}_1^{(d)}]$, where we assume the one-particle Hilbert space to be finite, $d$-dimensional.
We assume that the ground state is non-degenerate and that the respective $D(\vec{\lambda}(\kappa))$ is analytical in $\kappa$, at least in a neighbourhood of $\kappa=0$. Since $D(\vec{\lambda}(0))=0$, following from the fact that $\vec{\lambda}(0)=(1,\ldots,1,0,\ldots,0)$ is a vertex of the polytope (\emph{``Hartree-Fock point''}), and $D(\vec{\lambda}(\kappa))\geq 0$ for all $\kappa$, the linear order of $D(\vec{\lambda}(\kappa))$ must vanish and therefore the leading order correction is quadratic. Determining this second order term by exploiting conventual perturbation theory is a rather lengthy exercise. In a first step, one would need to determine $\ket{\Psi(\kappa)}$ up to order $\kappa^2$. Then, one would need to determine the one-particle reduced density operator,
\begin{eqnarray}\label{PT1RDM}
\rho_1(\kappa) &\equiv& N\Tr_{N-1}[\ket{\Psi(\kappa)}\!\bra{\Psi(\kappa)}] \nonumber \\
&\equiv& \sum_{j=1}^d \lambda_j(\kappa) \ket{\varphi_j(\kappa)}\!\bra{\varphi_j(\kappa)}\,,
\end{eqnarray}
of $\ket{\Psi(\kappa)}$ by tracing out $N-1$ fermions. Recall that the natural occupation numbers $\lambda_j(\kappa)$ shall be ordered non-increasingly and the respective \emph{natural orbitals} $ \ket{\varphi_j(\kappa)}$ are uniquely defined as long as the natural occupation numbers are non-degenerate.
Finally, one would need to perform again second order perturbation theory, this time on $\rho_1(\kappa)$, to determine the natural occupation numbers up to corrections of $\mathcal{O}(\kappa^2)$. This last step is particularly challenging since it involves degenerate unperturbed eigenvalues. Indeed, one has $\mbox{spec}\big(\rho_1(0)\big)=(1,\ldots,1,0,\ldots,0)$. In the following we therefore elaborate on a self-consistent perturbation theory, which may simplify the task quite a lot.

\subsection{Self-consistent perturbation theory}\label{sec:ptselfcon}
The crucial point of our perturbation theoretical approach is that it exploits a self-consistent expansion of $\ket{\Psi(\kappa)}$. We namely expand $\ket{\Psi(\kappa)}$ for each coupling $\kappa$ as a linear combination of the Slater determinants built from its own natural orbitals \cite{Lowdin},
\begin{equation}\label{PTselfcon}
\ket{\Psi(\kappa)} = \sum_{\bd{i}} c_{\bd{i}}(\kappa)\, \ket{\bd{i}(\kappa)}\,.
\end{equation}
Here, we use the shorthand notation
\begin{equation}\label{PTSD}
\ket{\bd{i}(\kappa)}\equiv \ket{\varphi_{i_1}(\kappa),\ldots,\varphi_{i_N}(\kappa)}
\end{equation}
for the respective $N$-particle Slater determinant constructed from the $N$ natural orbitals $\ket{\varphi_{i_1}(\kappa)},\ldots, \ket{\varphi_{i_N}(\kappa)}$ and $\bd{i}\equiv (i_1,\ldots,i_N)$. Moreover, $\ket{\Psi(\kappa)}$ is normalized to unity, $\langle \Psi(\kappa)\ket{\Psi(\kappa)}=1$. The same shall hold for the natural orbitals and therefore also for all $\ket{\bd{i}(\kappa)}$. The expansion \eqref{PTselfcon} is self-consistent in the sense that the respective expansion coefficients $c_{\bd{i}}(\kappa)$ fulfill self-consistency conditions, ensuring that the respective one-particle reduced density operator \eqref{PT1RDM} is diagonal with respect to its own natural orbitals and that the natural occupation numbers are ordered non-increasingly.

As discussed in Appendix \ref{app:pt}, the self-consistent expansion \eqref{PTselfcon} implies several convenient structural simplifications. First of all, a compact expression follows for the distance $D(\vec{\lambda}(\kappa))$ for all $\kappa$. To explain this, we employ second quantization using the natural orbitals of $\ket{\Psi(\kappa)}$ as reference basis. We can then express the natural occupation numbers
as particle number expectation values, $\lambda_j(\kappa)= \bra{\Psi(\kappa)} f^\dagger(\varphi_j(\kappa))f(\varphi_j(\kappa))\ket{\Psi(\kappa)}$, where $f^\dagger(\chi)$ and $f(\chi)$ create and annihilate a fermion in the state $\ket{\chi}$. Hence, by introducing for any GPC \eqref{GPC} the respective operator,
\begin{equation}
\hat{D}_{\Psi(\kappa)} \equiv D\big((\hat{n}_j(\kappa))_{j=1}^d\big)\,,
\end{equation}
$\hat{n}_j(\kappa)\equiv f^\dagger(\varphi_j(\kappa))f(\varphi_j(\kappa))$, we obtain (see also Ref.~\cite{CSQuasipinning})
\begin{equation}
D(\vec{\lambda}(\kappa)) = \bra{\Psi(\kappa)} \hat{D}_{\Psi(\kappa)}\ket{\Psi(\kappa)}\,.
\end{equation}
Moreover, one observes the elementary identity $\bra{\bd{j}(\kappa)} \hat{D}_{\Psi(\kappa)} \ket{\bd{i}(\kappa)}=\delta_{\bd{j},\bd{i}}D(\vec{e}_{\bd{i}})$, where $\vec{e}_{\bd{i}}$ denotes a vector with entries $\left(\vec{e}_{\bd{i}}\right)_k=0,1$ depending on whether $k$ is contained in $\bd{i}$ (1) or not (0).
%For instance for the configuration $\bd{i}=(2,5,6)$ belonging to the setting $(N,d)=(3,6)$ one would have $\vec{e}_{\bd{i}}=(0,1,0,0,1,1)$.
This identity then immediately leads to the compact expression
\begin{equation}\label{Dselfcon}
D(\vec{\lambda}(\kappa)) = \sum_{\bd{i}} |c_{\bd{i}}(\kappa)|^2\, D(\vec{e}_{\bd{i}})\,.
\end{equation}
It is remarkable that the right-hand side depends on $\kappa$ only via the coefficient functions $c_{\bd{i}}(\kappa)$. Moreover, \eqref{Dselfcon} holds not only for GPCs but for any function $D$ linear in the natural occupation numbers $\lambda_1,\ldots,\lambda_d$.
In the following, we consider only such GPCs which contain the Hartree-Fock point $\vec{\lambda}_{HF}\equiv (1,\ldots,1,0,\ldots,0)$, i.e.~constraints that are saturated for zero interaction, $\kappa \rightarrow 0$. By resorting to a perturbation theoretical expansion of the $c_{\bd{i}}(\kappa)$ and the natural orbitals $\ket{j(\kappa)}\equiv \ket{\varphi_j(\kappa)}$ one eventually finds (see derivation in Appendix \ref{app:pt})
\begin{eqnarray}\label{ptD}
D(\vec{\lambda}(\kappa))&=& \kappa^2\,\sum_{\bd{i}\in \mathcal{I}_2}\, \big|\bra{\bd{i}(0^+)}\big(\hat{H}^{(0)}-E^{(0)}\big)^{-1}\hat{V}\ket{\bd{i}_0}\big|^2 D(\vec{e}_{\bd{i}})\nonumber \\
&&+ \mathcal{O}(\kappa^3)
\end{eqnarray}
%\begin{equation}
%D(\vec{\lambda}(\kappa))= \kappa^2\,\sum_{\bd{i}\in \mathcal{I}_2}\, \left|\frac{\bra{\bd{i}(0^+)}\hat{V}\ket{\bd{i}_0}}{E_{\bd{i}}^{(0)}-E_{\bd{i}_0}^{(0)}}\right|^2 D(\vec{e}_{\bd{i}})+ \mathcal{O}(\kappa^3)
%\end{equation}
Here, the sum restricts to $\mathcal{I}_2$, the set of configurations $\bd{i}$ differing by exactly two orbital indices from $\bd{i}_0\equiv (1,2,\ldots,N)$ and $E^{(0)}$ denotes the energy of the unperturbed ground state $\ket{\Psi(0)}=\ket{\bd{i}_0}$. Since $\rho_1(0)$ has a degenerate spectrum, the expression \eqref{ptD} involves the \emph{adapted} unperturbed natural orbitals $\ket{j(0^+)}\equiv \ket{\varphi_j(0^+)}$ which do in general not coincide with the eigenstates $\ket{\varphi_j}$ of the one-particle Hamiltonian $\hat{H}^{(0)}$, $\ket{\varphi_j(0^+)}\neq \ket{\varphi_j}$. The adapted natural orbitals, formally defined via the limit process $\kappa \rightarrow 0^+$, can be determined without much computational effort \cite{Davidson76}. Hence, Eq.~\eqref{ptD} defines a striking connection between the pair interaction of the physical system and the quantum information theoretical quantity $D(\vec{\lambda})$, quantifying the absolute influence of the fermionic exchange symmetry on the one-particle picture.

An additional comment is in order, emphasizing the significance of result \eqref{ptD}.
In general, after determining the adapted natural orbitals, one could implement a unitary basis set transformation from the one-particle eigenstates $\ket{\varphi_j}$ of $\hat{H}^{(0)}$ to those adapted states $\ket{\varphi_j(0^+)}$. This would change $\hat{V}$ to another pair interaction $\hat{V}'$ and also $\hat{H}^{(0)}$ to another one-particle Hamiltonian with same energy spectrum . The respective expression \eqref{ptD} would then help to understand the mechanism behind quasipinning: The form of $\hat{V}'$ is related in the simplest possible way to the leading order of the distance $D(\vec{\lambda}(\kappa))$ of $\vec{\lambda}(\kappa)$ to the polytope boundary.
%It is also worth noticing that for \emph{all} translationally invariant $1$-band lattice models the adapted natural orbitals are known from the very beginning. They are given by the momentum states (multiplied by a spin state in case of electrons). Of course, one still would need to determine the correct ordering of the natural occupation numbers (momentum occupation numbers) which can be done easily via \eqref{ptD} as well: $\lambda_j(\kappa)$ is obtained by replacing the GPC $D$ by $P_j(\lambda_1,\ldots,\lambda_d)=\lambda_j$ (and adding a $0$-th order term $1$ for $j\leq N$.

%
%and $E_{\bd{i}}^{(0)}$ denotes the unperturbed energy for the unperturbed eigenstate $\ket{\bd{i}(\kappa=0)}\equiv\ket{\bd{i}}$. For instance, for an external harmonic trapping potential, one would have $E_{\bd{i}}^{(0)}=\sum_{k=1}^N \big(i_k+\frac{1}{2}\big)$. In case the pair interaction $\hat{V}$ takes the form of \eqref{ham}, analytical expressions can be found for the matrix elements $\bra{\bd{i}}\hat{V}\ket{\bd{j}}$ ???
%By referring to the form $\hat{V}=\sum_{j<k}\hat{V}_{j k}^{(2)}$, where $\hat{V}_{j k}$ acts non-trivially only on the Hilbert space of the $j$-th and $k$-th fermion one finds
%\begin{equation}
%\bra{i_1,i_2}\hat{V}\ket{j_1,j_2}= \Gamma_{\frac{s+1}{2}}
%\end{equation}

\section{Numerically exact treatment through DMRG}\label{sec:DMRG}
To apply DMRG in the context of continuously confined fermions we use its quantum chemical version (QC-DMRG) \cite{White-1999} adapted to spinless fermions and express Hamiltonian \eqref{ham} in second quantization. As truncated reference basis we choose the first $d$ oscillator states $\ket{\varphi_j}$ of the external harmonic trap (now playing the role of the ``lattice sites'' in standard DMRG).  The Hamiltonian then takes the form
\begin{equation}
\hat{H} = \sum_{i,j=0}^{d-1} h_{i;j} c_i^\dagger c_j + \frac{1}{2} K \!\!\!\!\sum_{i_1,i_2,j_1,j_2=0}^{d-1} \!\!\!V_{i_1 i_2;j_1,j_2} c_{i_1}^\dagger c_{i_2}^\dagger c_{j_2} c_{j_1}\,,
\end{equation}
where (recall $m\equiv \omega \equiv \hbar\equiv 1$)
\begin{equation}
h_{i;j} \equiv \bra{\varphi_i}\,\,\hat{p}_1^2/2+ \hat{x}_1^2/2 \,\,\ket{\varphi_j}  = \Big(j+\frac{1}{2}\Big)\,\delta_{i j}
\end{equation}
and
\begin{eqnarray}\label{eq:2matrix}
V_{i_1 i_2;j_1,j_2} &\equiv& \bra{\varphi_{i_1}}\!\otimes\!\bra{\varphi_{i_2}} \,|\hat{x}_1-\hat{x}_2|^s\,\ket{\varphi_{j_1}}\!\otimes\!\ket{\varphi_{j_2}}
\end{eqnarray}
In a tedious derivation --- being part of the long-term establishment of a DMRG scheme for systems of continuously confined fermions \cite{CSOLDMRG} --- one can determine an analytical expression for the two-particle matrix elements
\begin{widetext}
\begin{eqnarray}\label{eq:2matrix6}
V_{i_1 i_2;j_1,j_2}&=& 2^{s/2} \frac{(-1)^{i_2+j_2}}{\sqrt{i_1! i_2! j_1! j_2!}}  \sum_{m_1=0}^{\min{(i_1,j_1)}}\,\,\sum_{m_2=0}^{\min{(i_2,j_2)}} \binom{i_1}{m_1}\binom{j_1}{m_1}\binom{i_2}{m_2}\binom{j_2}{m_2} \, m_1!\, m_2!\, \mathcal{J}_{i_1+i_2+j_1+j_2-2m_1-2m_2}\,,
\end{eqnarray}
\end{widetext}
assuming $i_1+i_2+j_1+j_2$ to be even since otherwise $V_{i_1 i_2;j_1,j_2}$ vanishes (recall the invariance of the pair interaction under spatial inversion). Here, $\mathcal{J}_{k} \equiv \Gamma(\frac{s+1}{2}) \frac{2^{-k/2}}{\sqrt{\pi}}\prod_{j=0}^{k/2-1} (s-2j)$ and $\Gamma$ denotes the Gamma function.

By choosing sufficiently large bases of up to $d=80$ orbitals, we ensure the convergence of both, the variational energy and the natural occupation numbers on at least ten digits. In particular, we use the dynamic block state selection (DBSS) procedure \cite{Legeza-2003a,Legeza-2004b} to reach a threshold accuracy of $10^{-13}$ in the energy. We also invoked the dynamically extended active space (DEAS) procedure \cite{Legeza-2003b} with a minimum number of block states set to $M=1024$ to guarantee fast and stable convergence during the initialization sweep of the DMRG. The residual error threshold for the respective L\'anczos and Davidson diagonalization procedure is set to $10^{-13}$.

In the left panel of Fig.~\ref{fig:gs} we illustrate the convergence of our approach by comparing our variational energy to the exact one for the analytically solvable harmonic case \cite{harmOsc2012}, i.e., $s=2$,  for the fixed coupling $\kappa=1$ and $N=2,3,4,5$ fermions. Convergence of the energy on more than ten digits is achieved in our approach by choosing $d$ sufficiently large.
\begin{figure}[htb]
\includegraphics[width=0.38\columnwidth,height=4.5cm]{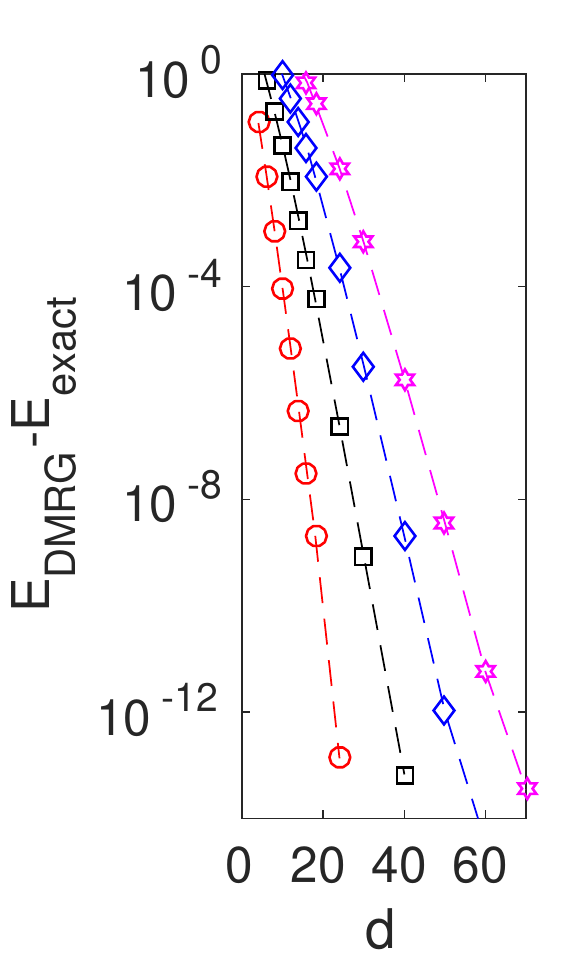}
\includegraphics[width=0.60\columnwidth,height=4.5cm]{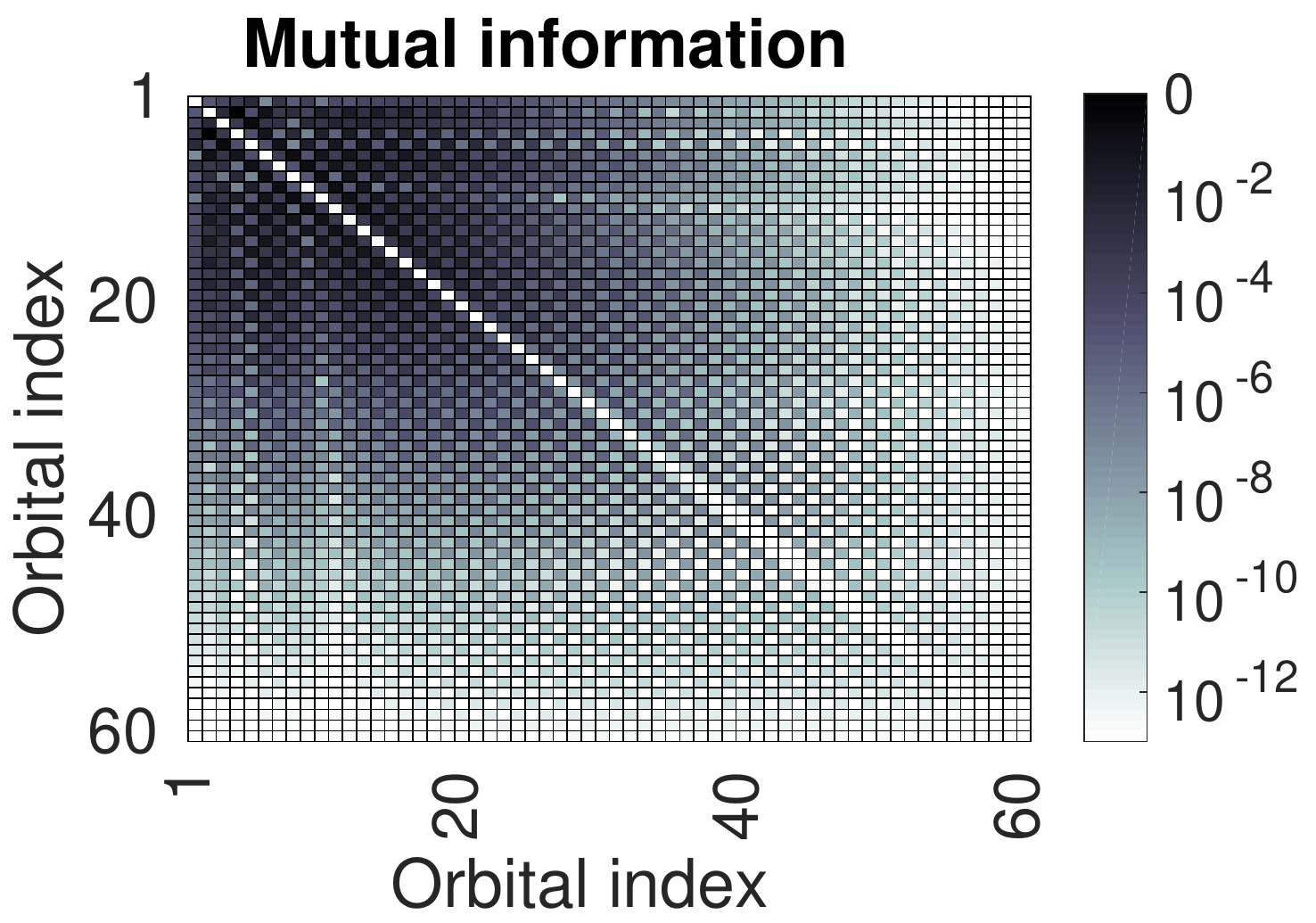}
\caption{Left: Absolute error of the DMRG ground state energy for $s=2$ (harmonic interaction), $\kappa= 1$ as a function of the basis set size $d$. Cases $N=2,3,4,5$ are represented by symbols
{\color{red}{$\pmb{\ocircle}$}},
{\color{black}{$\pmb{\Box}$}},
{\color{blue}{$\pmb{\lozenge}$}},
{\color{magenta}{$\pmb{\davidsstar}$}}, respectively and dashed lines emphasize the exponentially fast convergence in $d$. Right: Two-orbital mutual
information $I_{i,j}$ for any two orbitals $i,j=1,2,\ldots,60$ for same physical system and $N=5$.}
\label{fig:gs}
\end{figure}
To illustrate the need for large basis sets from a different perspective, we present in the right panel of Fig.~\ref{fig:gs} for the same physical system the mutual information $I_{i,j}$ for any two orbitals $i,j=1,2,\ldots,60$. Recall that $I_{i,j}$ quantifies the correlation between orbitals $i,j$ since it quantifies the extra information contained in the orbital reduced density operator $\rho_{ij}$ beyond the information already contained in both single orbital reduced density operators $\rho_i,\rho_j$, i.e.~in $\rho_i\otimes\rho_j$ \cite{Legeza-2003b,Legeza-2006a,Rissler-2006,Barcza-2011,Szalay-2015}. While the orbitals around the Fermi level are apparently the most active ones, it can also be inferred from Fig.~\ref{fig:gs} that large basis sets, $d\gg N$, are required to cover also dynamical correlation up to high precision.
%\begin{figure}[htb]
%\includegraphics[width=0.95\columnwidth,height=4cm]{graphics/mutualI}
%\caption{Illustration of the two-orbital mutual
%information $I_{i,j}$ for any two orbitals $i,j=1,2,\ldots,60$ for $N=5$ fermions and $\kappa=1, s=2$ on a logarithmic scale.
%Orbitals are arranged for better visibility into two circles (one for odd and one for even spatial parity).}
%\label{fig:I}
%\end{figure}

For a detailed presentation of the DMRG scheme which we developed to describe systems of continuously confined fermions and a comprehensive analysis of the entanglement structure of those systems we refer the reader to \cite{CSOLDMRG}.

\section{Results}\label{sec:sgrid}
For various numerically exact ground states calculated by DMRG we determine the corresponding one-particle reduced density matrices and diagonalize them numerically to obtain the natural occupation numbers $\lambda_k$. Since our high-precision approach involves large active spaces and since the GPCs are known so far only up to basis sets of size $d=12$ we resort to the concept of truncation \cite{CS2016a}. We perform the (quasi)pinning analysis in terms of a \emph{truncated} vector $\vec\lambda'$, obtained by discarding all NONs sufficiently close to 0 (and also 1).
To be more specific, we quantify quasipinning by the minimal $l_1$-distance $D$ of $\vec{\lambda}$ to the polytope boundary. We then reduce $N$ to $N'$ by ignoring eigenvalues close to $1$, and $d$ to $d'$ by also ignoring those close to $0$. The minimal distance $D'$ of $\vec\lambda'$ to the boundary of the polytope $\mathcal{P}'$ of $(N',d')$ coincides with $D$ in the full setting up to a truncation error $\varepsilon'$,
\begin{equation}\label{eq:trunc}
\big|D-D'\big| \leq \varepsilon'\equiv\sum_{j=1}^{N-N'}(1-\lambda_j)+\!\!\sum_{k=0}^{d-d'-N'+N-1}\!\!\lambda_{d-k}\,.
\end{equation}
Furthermore, since the polytope $\mathcal{P}$ is a subset of the Pauli simplex $1\geq \lambda_1\geq \ldots \geq \lambda_d\geq 0$,
quasipinning can be considered as \emph{non-trivial} only if the distance of $\vec{\lambda}$ to the polytope boundary $\partial P$ is much smaller than the distance of $\vec{\lambda}$ to the boundary of the Pauli simplex. This ``degree of non-triviality'' is quantified by the \emph{$Q$-parameter}  \cite{CSQ}.

\begin{figure}[h]
   \includegraphics[width=7.55cm]{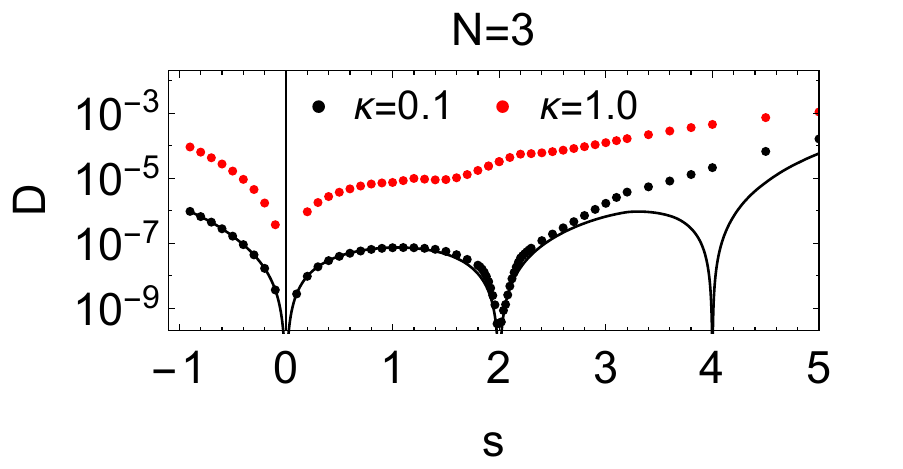}
   \includegraphics[width=7.55cm]{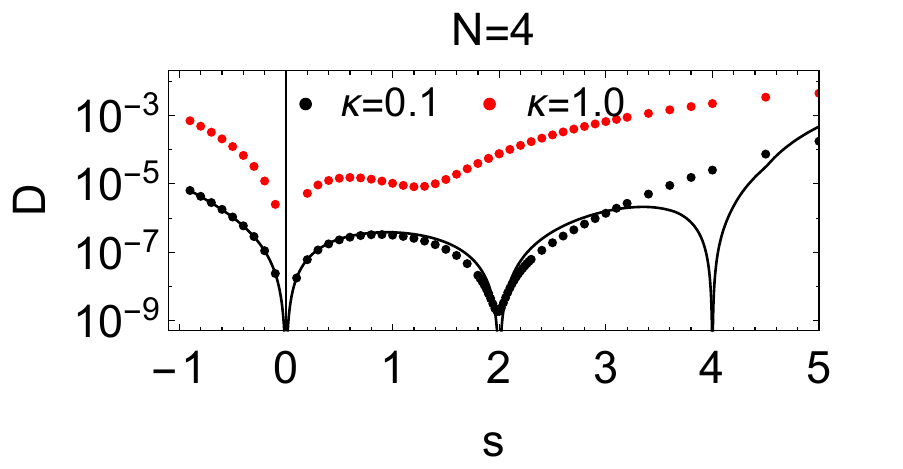}
   \includegraphics[width=7.55cm]{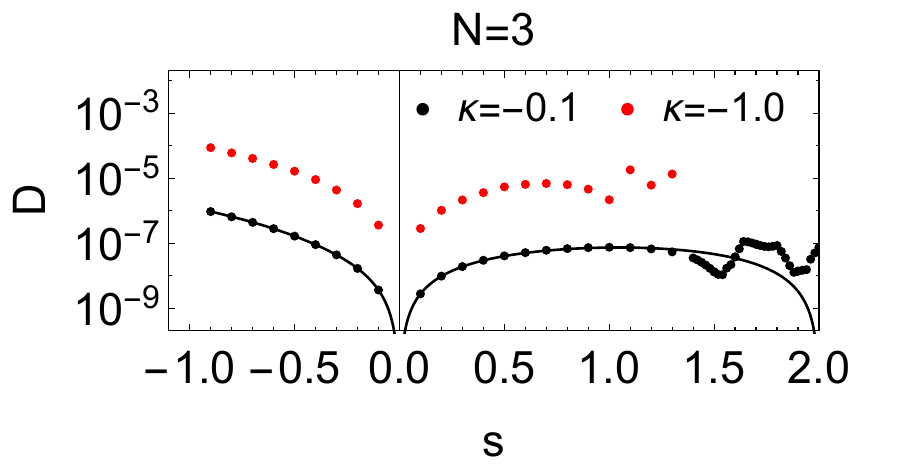}
   \includegraphics[width=7.55cm]{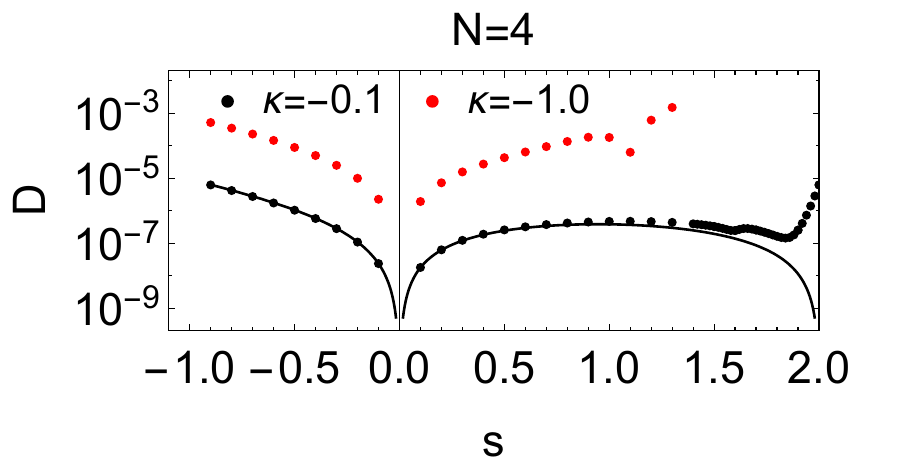}
    \caption{For the ground state of the Hamiltonian (\ref{ham}) we present for different coupling strengths $\kappa$ the $s$-dependence of the minimal distance $D$ of the vector $\vec{\lambda}$ of NON to the polytope boundary. The quadratic leading order of $D$ following from the perturbation theory is shown as solid line.}
    \label{fig:D}
\end{figure}
In Fig.~\ref{fig:D} we present the results for $D$ for a large grid of about 100 $s$-values for the case of $N=3,4$ fermions for weak coupling ($\kappa=\pm 0.1$) and medium coupling ($\kappa=\pm 1.0$). For the case of negative couplings we restrict to $s<2$ since the system has no bound ground state for $s>2$. For all grid points we apply the \emph{``concept of truncation''} \cite{CS2016a} as explained above. The respective truncation errors for the particle numbers $N=3,4$ and the chosen couplings $\kappa$ turn out to be negligibly small, i.e.~they are smaller than the width of the plot markers. Even for medium interaction strengths, those numerical results confirm that the GPCs have an \emph{absolute} relevance for all $s$-values since the distance $D$ of the NONs $\vec{\lambda}$ to the polytope boundary is much smaller than the length scale of the polytope given by $\mathcal{O}(1)$.
For the specific cases of a harmonic ($s=2$) and quartic ($s=4$) pair interaction the system exhibits an extremely strong quasipinning which is mainly due to the fact that the active space size shrinks for $s=2n, n\in \mathbb{N}$ \cite{CS2013,CSthesis,CS2016a}. In that context, the reader should note that $s=0$ represents a non-interacting system implying $D(\vec{\lambda}(\kappa))=0$ for all couplings.
A comparison of the results for small coupling ($|\kappa|=0.1$) with those for medium coupling ($|\kappa|=1.0$) shows that increasing the coupling
leads to a weakening of the quasipinning which is due to the increase of the total correlation. Nonetheless, given that $D(\vec{\lambda})\leq 10^{-3}$ for $|\kappa|=1.0$ for all $s$-values, the quasipinning for medium coupling is still quite strong.

Particularly remarkable is the fact that the unique nature of the extremely strong quasipinning in the neighbourhood of $s=2,4$ reduces a lot as one increases the coupling from very-small ($|\kappa|\ll 1$) to small ($|\kappa|=0.1$) and eventually medium coupling ($|\kappa|=1.0$). In addition, it is also remarkable that for $\kappa=\pm0.1$ the perturbation theoretical results for $D(\vec{\lambda})$ agree with the numerically exact DMRG results almost perfectly for $-1<s<2.5$ and $-1<s<1.3$ in case of $\kappa>0$ and $\kappa<0$, respectively. For $s>3$ the numerically exact results do not agree with the second-order perturbation theoretical results. This is due to the fact that for such extreme pair interactions, the coupling $\kappa=0.1$ is not yet small enough and thus  higher orders ($\kappa^3$) strongly affect the behaviour of $D(\vec{\lambda(\kappa)})$. Much better agreement between the numerical results and the perturbation theoretical results can be found for $s>3$, however, by considering smaller couplings as, e.g., $\kappa=0.01$ (not presented here).

\begin{figure}[h]
   \includegraphics[width=7.55cm]{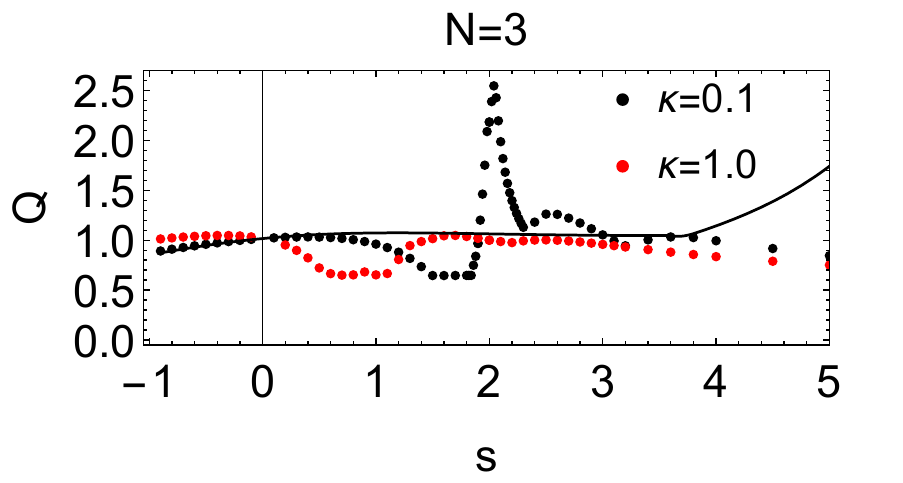}
   \includegraphics[width=7.55cm]{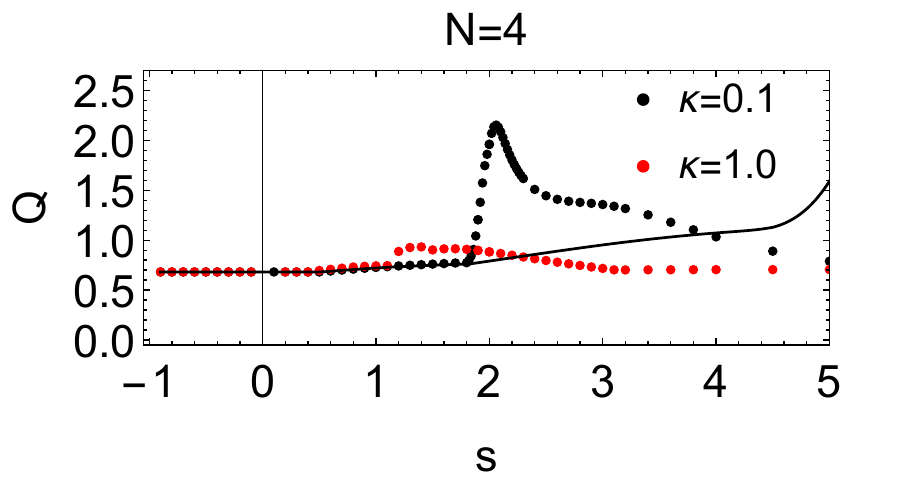}
   \includegraphics[width=7.55cm]{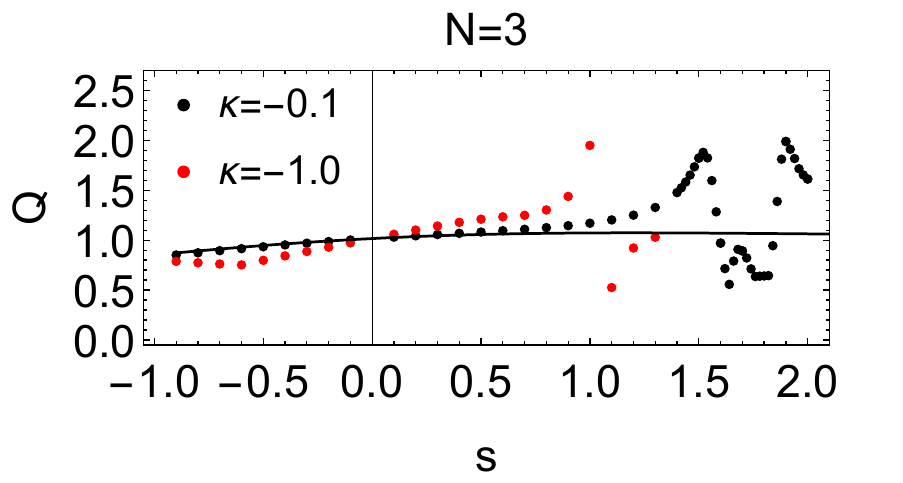}
   \includegraphics[width=7.55cm]{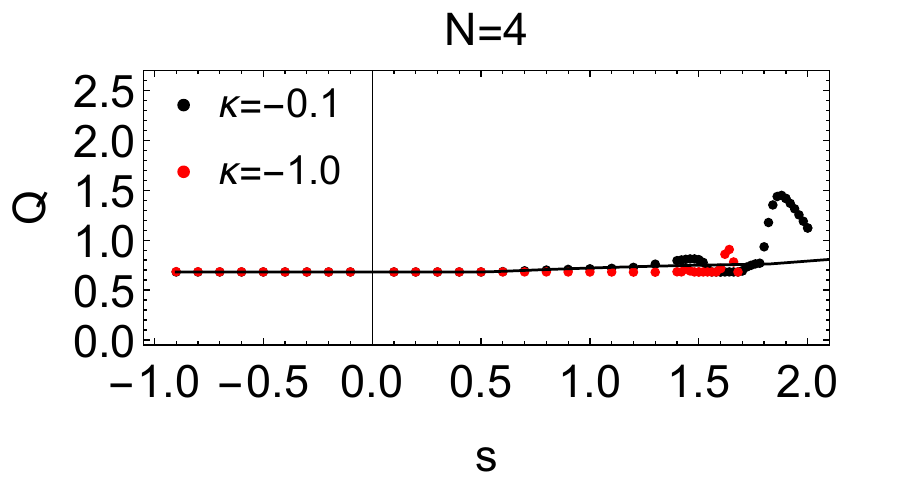}
    \caption{For the ground state of the Hamiltonian (\ref{ham}) we present for different coupling strengths $\kappa$ the $s$-dependence of the $Q$-parameter. The leading order of $Q$ following from the perturbation theory is shown as solid line.}
    \label{fig:Q}
\end{figure}

While Fig.~\ref{fig:D} confirms the relevance of the fermionic exchange symmetry on the one-particle picture, it is important to also understand to what extent this relevance needs to be assigned to Pauli's original principle. The potential significance of the GPCs beyond the already well-established relevance of Pauli's exclusion principle is quantified by the $Q$-parameter \cite{CSQ}. Recall that a value $Q_j(\vec{\lambda})$ of the $Q$-parameter for the $j$-th GPC, $D_j\geq 0$, means that this GPC is $10^{Q_j(\vec{\lambda})}$ stronger saturated than what one could expect from the approximate saturation of some Pauli constraints. The results for the overall $Q$-parameter $Q=\max_{j}(Q_j)$ are presented in Fig.~\ref{fig:Q}. First, we infer that the GPCs have a \emph{non-trivial} relevance for all $s$-values and all considered $\kappa$-values, i.e.~some GPCs are always saturated by a factor of about $10$ stronger than Pauli's exclusion principle constraints would suggest. The singular behaviour of $D(\vec{\lambda})$ at $s=2,4$ is also present in the $Q$-parameter, at least for $s=2$. For $s=4,6,\ldots$ the present DMRG results cannot resolve these singularities. Yet, we are convinced that the non-trivial character of the quasipinning is particularly pronounced at all positive even integer values of $s$, and in the respective vicinities.
Furthermore, the perturbation theory (shown as solid line) shows that the leading (zeroth) order $Q^{(0)}$ of $Q(\vec{\lambda}(\kappa))$ is approximately constant for almost all $s$ in complete contrast to the harmonic case. Indeed, for $s=2$ and probably also for further even-integer $s$-values above $s=2$ the $Q$-parameter diverges as the coupling $\kappa$ tends to zero (and thus does not allow for a perturbation theoretical expansion of $Q(\vec{\lambda}(\kappa))$) around $\kappa=0$, i.e., $Q(\vec{\lambda}(\kappa))$ is not analytical at $\kappa=0$.

\section{Summary and conclusion}\label{sec:concl}
Conclusive results for a harmonic model system \cite{CS2013,Ebler,CSthesis,CS2016a} provided first evidence that the generalized Pauli constraints (GPC) would have a tremendous physical relevance since they were found to be saturated up to a very small correction of the order eight in the coupling, $D\propto \kappa^8$. In the present work, we explore whether those seminal findings --- the presence of such extremely strong quasipinning in ground states up to medium interaction strength --- also hold for various other few-fermion systems.  Answering this question is, however, highly challenging since the \emph{exact} description of \emph{interacting} fermions in the continuum can be considered as one of the hardest problems in physics. Even worse, as recent studies of atoms and molecules based on rather small active spaces have revealed, quasipinning found in those systems can be artificial in the sense that it is a mere consequence of the orbital and spin symmetries \cite{BenavQuasi2,CS2015Hubbard,CSthesis,CBthesis}. In our work, we largely avoided those bottlenecks by considering one-dimensional systems of spinless fermions interacting by a general pair potential of the form $\kappa |\hat{x}_i-\hat{x}_j|^s$ confined by a harmonic external potential. We resorted to the DMRG approach in the context of continuously confined fermions to determine the numerically exact ground states of those systems up to medium coupling $\kappa$ for the exemplary cases of $N=3,4$ fermions and for about 100 different pair interactions $-1<s\leq s_{max}$, where $s_{max}=5$ for attractive coupling ($\kappa>0$) and $s_{max}=2$ for repulsive coupling ($\kappa<0$). By choosing sufficiently large basis sets of up to 80 orbitals we ensured the convergence on more than ten digits of both, the variational energy and the natural occupation numbers $\vec{\lambda}\equiv (\lambda_i)$.
%This may be contrasted with the best numerical studies of the GPCs in atoms and molecules which to our best knowledge reached a convergence of the variational energy on no more than two digits (see, e.g., \cite{BenavLiQuasi}).

Our numerically exact analysis confirms the original expectation: The GPCs are indeed universally relevant in the sense that they are  approximately saturated regardless of the type of pair interaction. Even for medium interaction strength and almost the whole regime $-1<s<5$ of considered pair interactions, we found a quasi-saturation $D(\vec{\lambda})\leq 10^{-3}$ of the GPCs. This provides further evidence that such quasipinning has its origin in the conflict between energy minimization and fermionic exchange symmetry which is present in all systems of continuously confined fermions. To distinguish between genuine and trivial quasipinning we used the $Q$-parameter \cite{CSQ}. The comprehensive analysis eventually confirms that the quasipinning by the GPCs is not primarily a result of the approximate saturation of Pauli's exclusion principle constraints $0\leq \lambda_j \leq 1$.

To shed more light on the weak coupling regime, $|\kappa| \ll 1$, we developed a self-consistent perturbation theory which is based on an expansion of the $N$-fermion quantum states in Slater determinants built from its own natural orbitals. The respective analytical results agree quite well for most of the $s$-regime for coupling $\kappa=\pm 0.1$ with the numerical results. Furthermore, the analytical results also elucidate the singular behaviour found for the specific values $s=2,4,\ldots$, resulting in an even stronger, rather extreme quasipinning compared to generic $s$-values for weak couplings. In particular, the perturbation theory provides a conceptually important insight into quasipinning with consequences for the related terminology: Quasipinning in generic systems is quite strong despite the fact that the leading order correction is only quadratic, $D(\vec{\lambda}(\kappa))= c_2 \kappa^2$. It is thus the respective prefactor ($c_2$) rather than the exponent of the leading order of $D(\vec{\lambda}(\kappa))$ which defines the strength of quasipinning. Only for the specific case $s=2$ and probably some further even-integer $s$-values, the quadratic and further higher orders in the expansion of $D(\vec{\lambda}(\kappa))$ vanish rigorously.

All these results, presented in our work for $N=3,4$ fermions, hold qualitatively also for larger particle numbers.

Due to the remarkable implication of (quasi)pinning for the structure of the many-fermion wave functions \cite{Stability}, our findings emphasize again the potential significance that GPCs may have in few-body quantum systems, particularly for Multi-Configurational Self-Consistent Field (MCSCF) ansatzes and in Reduced Density Matrix Functional Theory (RDMFT). Such applications of the GPCs would require, however, that the recent development \cite{Kly2,Kly3,Altun,SK11,SOK2012,ST13,MOS13,HKS13,SOK14,Alex,MS15,symplecticreview,MS17} in quantum information sciences and mathematical physics yields more efficient algorithms for the calculations of the GPCs for larger active spaces.

%Furthermore, we believe that our analysis of the non-trivial relevance of the GPCs provides a rather good idea of the \emph{maximal} possible genuine relevance that GPCs can have: One the one hand, the choice of the systems \eqref{ham} ensured that possible (quasi)pinning is neither an artifact of orbital-symmetries nor of spin-symmetries \cite{BenavQuasi2,CS2015Hubbard,CBthesis}. On the other hand, freezing degenerate angular and spin degrees of freedom and considering a steep external trap enhances the conflict between energy minimization and fermionic exchange symmetry and thus enhances quasipinning \cite{CS2016a}.

\begin{acknowledgements}
%We are grateful V.\hspace{0.5mm}Vedral and Z.\hspace{0.5mm}Zimbor\'as for helpful discussions.
We acknowledge financial support from the Oxford Martin Programme on Bio-Inspired Quantum Technologies and the UK Engineering and Physical Sciences Research Council (Grant EP/P007155/1) (C.S.), the Hungarian National Research, Development and Innovation Office (NKFIH) through Grant No. K120569 and the Hungarian Quantum Technology National Excellence Program (Project No. 2017-1.2.1-NKP-2017-00001) (O.L.)
\end{acknowledgements}

\appendix

\section{Details of the self-consistent perturbation theory}\label{app:pt}
In this appendix we present all technical details of the self-consistent perturbation theory proposed in Sec.~\ref{sec:ptselfcon}.
%which proves highly convenient in describing the structure of ground states  in the regime of not too strong interactions.

Let us consider an $N$-particle Hamiltonian of the general form
\begin{equation}\label{appPTham}
\hat{H}(\kappa) = \hat{H}^{(0)}+\kappa \hat{V}\,,
\end{equation}
acting on the $N$-fermion Hilbert space $\wedge^N[\H_1^{(d)}]$, where $\H_1^{(d)}$ denotes the underlying one-particle Hilbert space of dimension $d$. Here, $\hat{H}^{(0)}$ is a general one-particle Hamiltonian including the kinetic energy and the external potential and $\hat{V}$ is a pair interaction.
We assume that the ground state $\ket{\Psi(\kappa)}$ of the Hamiltonian \eqref{appPTham} depends analytically on $\kappa$, at least in a neighbourhood of $\kappa=0$, and thus allows us to study it by perturbation theoretical means around $\kappa=0$.

We expand $\ket{\Psi(\kappa)}$ self-consistently according to (\ref{PTselfcon}). The respective natural orbitals $\ket{\varphi_j(\kappa)}\equiv \ket{j(\kappa)}$ follow from the one-particle reduced density operator \eqref{PT1RDM} obtained after tracing out $N-1$ fermions.
%The crucial point about our perturbation theoretical approach is that it exploits a self-consistent expansion of $\ket{\Psi(\kappa)}$. We namely expand $\ket{\Psi(\kappa)}$ for each coupling $\kappa$ as a linear combination of the Slater determinants built from its own natural orbitals,
%\begin{equation}\label{PTselfcon}
%\ket{\Psi(\kappa)} = \sum_{\bd{i}} c_{\bd{i}}(\kappa)\, \ket{\bd{i}(\kappa)}\,.
%\end{equation}
%Here, we use the shorthand notation
%\begin{equation}\label{PTSD}
%\ket{\bd{i}(\kappa)}\equiv \ket{\varphi_{i_1}(\kappa),\ldots,\varphi_{i_N}(\kappa)}
%\end{equation}
%for the respective $N$-particle Slater determinant constructed from the $N$ natural orbitals $\ket{\varphi_{i_1}(\kappa)},\ldots, \ket{\varphi_{i_N}(\kappa)}$. The expansion \eqref{PTselfcon} is self-consistent in the sense that the respective expansion coefficients $c_{\bd{i}}(\kappa)$ fulfill self-consistency conditions, ensuring that the respective $1$-particle reduced density operator \eqref{PT1RDM} is diagonal with respect to its own natural orbitals and that the natural occupation numbers are ordered non-increasingly.
This self-consistent expansion has a couple of convenient properties. To discuss them we expand
the coefficient functions $c_{\bd{i}}(\kappa)$,
\begin{equation}\label{PTc}
c_{\bd{i}}(\kappa) = c_{\bd{i}}^{(0)}+ \kappa\, c_{\bd{i}}^{(1)} + \mathcal{O}(\kappa^2)
\end{equation}
and the natural orbitals $\ket{j(\kappa)}$,
\begin{equation}\label{PTNO}
\ket{j(\kappa)} = \ket{j^{(0)}} + \kappa\, \ket{j^{(1)}} + \mathcal{O}(\kappa^2)\,.
\end{equation}
Here, the natural orbitals (and thus the Slater determinants $\ket{\bd{j}(\kappa)}$) shall be normalized to unity for all $\kappa$, $\langle j(\kappa)\ket{j(\kappa)}=1$. It is important to notice that the states $\ket{j^{(0)}}\equiv \lim_{\kappa\rightarrow 0^+} \ket{j(\kappa)}$ do in general not coincide with the one-particle eigenstates of the unperturbed Hamiltonian $\hat{H}^{(0)}$. Nonetheless, one has
\begin{eqnarray}\label{ptspan}
\mbox{span}\left(\{\ket{j}\}_{1\leq j \leq N}\right)&=&\mbox{span}\left(\{\ket{j^{(0)}}\}_{1\leq j \leq N}\right) \nonumber  \\
\mbox{span}\left(\{\ket{j}\}_{N+1\leq j \leq d}\right)&=&\mbox{span}\left(\{\ket{j^{(0)}}\}_{N+1\leq j \leq d}\right)
\end{eqnarray}
Concerning the total quantum state \eqref{PTselfcon}, the expansion in $\kappa$ reads
\begin{equation}\label{ptPsi}
\ket{\Psi(\kappa)} = \ket{\Psi^{(0)}}+ \kappa\, \ket{\Psi^{(1)}}+\mathcal{O}(\kappa^2).
\end{equation}
with the normalization condition $\langle \Psi(\kappa)\ket{\Psi(\kappa)}=1$. Since the ground state is assumed to be unique, we have
\begin{equation}\label{pt0c}
c_{\bd{i}}^{(0)}=\,\begin{cases}
1\,,\quad \mbox{if}\,\bd{i}=\bd{i}_0\\
0\,,\quad\mbox{otherwise}
\end{cases}
\end{equation}
where $\bd{i}_0\equiv (1,2,\ldots,N)$.
A first useful property of the self-consistent expansion \eqref{PTselfcon} is that configurations $\bd{i}$ differing by exactly one orbital index from the reference-configuration $\bd{i}_0\equiv(1,2,\ldots,N)$ contribute to $\ket{\Psi(\kappa)}$ with significantly reduced weight as stated in the following Lemma.
\begin{lem}\label{lem1}
Let $\ket{\Psi(\kappa)}$ be an $N$-fermion quantum state, analytical in $\kappa$, and normalized to unity, $\langle \Psi(\kappa)\ket{\Psi(\kappa)}=1$. We denote the overall first contribution in the series of the analytical coefficient functions $c_{\bd{j}}(\kappa)$ in the self-consistent perturbation expansion \eqref{PTselfcon} by $r$. Then, $c_{\bd{i}_0}(\kappa)=1 + \mathcal{O}(\kappa^{2r})$ and all configurations $\bd{i}$ differing from the reference configuration $\bd{i}_0\equiv(1,2,\ldots,N)$ by exactly one orbital index, $|\bd{i}\cap \bd{i}_0| =N-1$, contribute only with weight of the order $\mathcal{O}(\kappa^{2r})$.
\end{lem}
\begin{proof}
Consider a configuration $\bd{i}$ differing from the reference configuration $\bd{i}_0$ by the orbital index $\alpha>N$ which replaces the index $l\leq N$ in $\bd{i}_0$, i.e.~$\bd{i}=(\bd{i}_0 \cup\{\alpha\})\setminus \{l\}$ Due to the self-consistent character of the expansion \eqref{PTselfcon}, the one-particle reduced density operator $\rho_1(\kappa)$ is diagonal in the basis of its own natural orbitals. In particular,
\begin{eqnarray}\label{1rdmentry}
0&\overset{!}{=} &\bra{\alpha(\kappa)}\rho_1(\kappa)\ket{l(\kappa)} \nonumber \\
%&=& \bra{\Psi(\kappa)}f^\dagger(\varphi_l(\kappa)) f(\varphi_k(\kappa))\ket{\Psi(\kappa)}\nonumber \\
&=& \sum_{\bd{j}\ni l} c_{\bd{j}}^\ast(\kappa)\, c_{(\bd{j}\cup\{\alpha\})\setminus\{l\}}(\kappa) \,(-1)^{\#\{i\in \bd{j}|k<i<\alpha\}}\nonumber \\
&=& c_{\bd{i}_0}^\ast(\kappa)\, c_{\bd{i}}(\kappa)\,(-1)^{N-l+1} + \mathcal{O}(\kappa^{2r}) \,.
\end{eqnarray}
In the last line, we have used for $\bd{i}\neq \bd{i}_0$ that $c_{\bd{i}}(\kappa)=\mathcal{O}(\kappa^r)$, i.e., $r$ is the overall leading order in $\kappa$ of the series \eqref{PTc}. Finally, by using $c_{\bd{i}_0}(\kappa)=1+\mathcal{O}(\kappa^{2r})$ following from the normalization $1=\langle \Psi(\kappa)\ket{\Psi(\kappa)}=|c_{\bd{i}_0}(\kappa)|^2 + \mathcal{O}(\kappa^{2r})$ we then obtain from \eqref{1rdmentry} $c_{\bd{i}}(\kappa)=\mathcal{O}(\kappa^{2r})$.
\end{proof}
%The reader should not that Lemma \ref{lem1} also holds for any other normalizations condition for the analytical $\ket{\Psi(\kappa)}$.

It is worth noticing that for generic Hamiltonians of the form \eqref{PTham} the leading order corrections are of linear order, $r=1$.
For special cases, as, e.g., the Harmonium model defined by Hamiltonian \eqref{ham} with $s=2$, however, one can find $r>1$.

To determine the first order contributions $c_{\bd{i}}^{(1)}$ to the quantum state \eqref{PTselfcon} we study the respective time-independent Schr\"odinger equation
\begin{equation}\label{SEq}
E(\kappa) \ket{\Psi(\kappa)} = (\hat{H}^{(0)} +\kappa \hat{V}) \ket{\Psi(\kappa)}\,.
\end{equation}
In zeroth order, by expanding the energy according to
\begin{equation}\label{ptE}
E(\kappa)\equiv E^{(0)}+\kappa\, E^{(1)}+\mathcal{O}(\kappa^2)\,,
\end{equation}
we have
\begin{equation}\label{PT0K}
E^{(0)} \ket{\Psi^{(0)}} = \hat{H}^{(0)} \ket{\Psi^{(0)}}\,.
\end{equation}
This apparently yields (recall also \eqref{ptspan})
\begin{equation}\label{PT0Kresult}
\ket{\Psi^{(0)}}=\ket{\bd{i}_0}= \ket{\bd{i}_0(0^+)}\,,\quad E^{(0)}=E_{\bd{i}_0}\,,
\end{equation}
where for any configuration $\bd{i}\equiv(i_1,\ldots,i_N)$, $E_{\bd{i}}$ denotes the unperturbed energy of the respective Slater determinant $\ket{\bd{i}}$. $E_{\bd{i}}$  is nothing else than just the sum of the single-particle energies of the orbitals $\ket{i_1},\ldots \ket{i_N}$.
Moreover, we have indeed $\ket{\bd{i}_0}= \ket{\bd{i}_0(0^+)}$ (up to a phase which we can set to zero) which follows from \eqref{ptspan}.

In linear order, \eqref{SEq} leads to
\begin{equation}\label{PT1K}
E^{(1)}\ket{\Psi^{(0)}} + E^{(0)}\ket{\Psi^{(1)}} = \hat{H}^{(0)} \ket{\Psi^{(1)}} + \hat{V} \ket{\Psi^{(0)}}\,.
\end{equation}
By projecting \eqref{PT1K} onto $\ket{\Psi^{(0)}}=\ket{\bd{i}_0(0^+)}=\ket{\bd{i}_0}$ and using \eqref{PT0K} we find
\begin{equation}\label{PT0Kresult1}
E^{(1)} = \bra{\bd{i}_0}\hat{V}\ket{\bd{i}_0}\,.
\end{equation}
On the orthogonal complement of $\ket{\bd{i}_0(0^+)}$ we can invert the operator $E^{(0)}-\hat{H}^{(0)}$. Thus we obtain from \eqref{PT1K} restricted to $\mbox{span}\big(\big\{\ket{\bd{i}(0^+)}\big\}_{\bd{i}\neq \bd{i}_0}\big)$
\begin{equation}\label{pt1KPsi1}
\ket{\Psi^{(1)}} = \big(E^{(0)}-\hat{H}^{(0)}\big)^{-1}\hat{V} \ket{\Psi^{(0)}}\,.
\end{equation}
%we obtain for all $\bd{i}\neq \bd{i}_0$
%\begin{equation}\label{PT0Kresult2}
% E^{(0)}\langle \bd{i}(0^+)\ket{\Psi^{(1)}} = E_{\bd{i}} \langle \bd{i}(0^+)\ket{\Psi^{(1)}} + \bra{\bd{i}(0^+)}\hat{V} \ket{\Psi^{(0)}}
%\end{equation}
%and thus
%\begin{equation}
%\langle \bd{i}(0^+)\ket{\Psi^{(1)}} =\bra{\bd{i}(0^+)} \big(\hat{H}^{(0)}-E^{(0)}\big)^{-1}\hat{V} \ket{\Psi^{(0)}}\,.
%\end{equation}
Moreover, by comparing \eqref{PTc} and \eqref{PTNO} with \eqref{ptPsi} and using \eqref{pt0c} and Lemma \ref{lem1} we find
\begin{eqnarray}\label{pt1KPsi2}
\ket{\Psi^{(1)}} & =& \sum_{\bd{i}} \left[c_{\bd{i}}^{(1)}\, \ket{\bd{i}(0^+)} + c_{\bd{i}}^{(0)}\, \ket{\bd{i}(\kappa)}^{(1)}\right] \nonumber \\
&& \sum_{\bd{i}\in \mathcal{I}_{\geq 2}} c_{\bd{i}}^{(1)}\, \ket{\bd{i}(0^+)} + \ket{\bd{i}_0(\kappa)}^{(1)}\,.
\end{eqnarray}
Here, $\ket{\bd{i}(\kappa)}^{(1)}$ denotes the linear order of $\ket{\bd{i}(\kappa)}$ in $\kappa$ and $\mathcal{I}_{\geq 2}$ denotes the set of all configurations $\bd{i}=(i_1,\ldots,i_N)$ differing in at least two indices from $\bd{i}_0\equiv (1,\ldots,N)$.

Since $\langle \bd{i}(0^+)  \ket{\bd{i}_0(\kappa)}^{(1)}=0$ for all $\bd{i}$ differing from $\bd{i}_0$ by more than one orbital index, \eqref{pt1KPsi2} yields
\begin{equation}
\langle \bd{i}(0^+) \ket{\Psi^{(1)}} = c_{\bd{i}}^{(1)}\,, \quad \forall \bd{i} \in \mathcal{I}_{\geq 2}\,.
\end{equation}
Using this in combination with \eqref{pt1KPsi1} leads to
\begin{equation}\label{PT1Kc}
c_{\bd{i}}^{(1)} = \bra{\bd{i}(0^+)}\big(E^{(0)}-\hat{H}^{(0)}\big)^{-1}\hat{V} \ket{\bd{i}_0}
\end{equation}
for all $\bd{i} \in \mathcal{I}_{\geq 2}$. Finally, we observe that according to \eqref{PT1Kc}, $c_{\bd{i}}^{(1)}$ also vanishes in case $\bd{i}$ differs from $\bd{i}_0$ by more than two indices since $\hat{V}$ is a two-particle operator. Hence, \eqref{ptD} follows from \eqref{Dselfcon} by using \eqref{PT1Kc}.

%\section{Determining the adapted natural orbitals}\label{app:ptDav}
%We state the expression for the $1$-particle reduced density operator $\rho_1(\kappa)$ up to order $\kappa^2$ as derived in Ref.~\cite{Davidson76}. Based on this, one can determine without much computational effort the adapted natural orbital basis set $\ket{j(0^+)}\equiv \ket{\varphi_j(0^+)}$, $j=1,\ldots,d$.
%%\section{Convergence of DMRG}\label{app:converg}
%\section{Further numerical results}\label{app:results}

%%\section{Weak interaction regime}\label{app:weak}

\bibliography{GPCgs}

\end{document}